\documentclass[a4paper,12pt]{article}

\usepackage{amsthm}
\usepackage{amsmath}
\usepackage{amssymb}
\usepackage{graphicx}
\usepackage[round]{natbib}
\usepackage{refcheck}
\usepackage{bbm}
\usepackage{textcomp}

\usepackage{color}
\definecolor{shadecolor}{gray}{0.9}

\theoremstyle{plain}  
\newtheorem{thm}{Theorem}[section] 
\newtheorem{lem}[thm]{Lemma} 
\newtheorem{prop}[thm]{Proposition} 
\newtheorem*{cor}{Corollary} 

\theoremstyle{definition}

\newtheorem{exmp}{Example}[section]

\newcommand{\TVaR}{\mathrm{TVaR}}
\theoremstyle{remark} 
\newtheorem*{rem}{Remark}

\newcommand{\diff}{\,\mathrm{d}}
\newcommand{\E}{\mathbb{E}}

\newcommand{\sS}{\mathcal{S}}
\newcommand{\R}{\mathbb{R}}
\newcommand{\Q}{\mathbb{Q}}
\newcommand{\p}{\mathbb{P}}

\newcommand{\Ind}{\mathbbm{1}}
\newcommand{\A}{\mathcal{A}}
\newcommand{\F}{\mathcal{F}}

\DeclareMathOperator*{\essinf}{ess\,inf}

\begin{document}

\title{Risk Measures with the CxLS property}
\author{Freddy Delbaen, Fabio Bellini, Valeria Bignozzi \\ and Johanna F.~Ziegel}
\maketitle
\begin{abstract} In the present contribution we  characterize law determined convex risk measures that have convex level sets at the level of distributions. 
By relaxing the  assumptions in \cite{Weber2006}, we show that these risk measures can be identified with a class of generalized shortfall risk measures. As a direct consequence, we are able to extend the results in \cite{Ziegel2013} and \cite{BelliniBignozzi2013} on convex elicitable risk measures and confirm that  expectiles are the  only elicitable coherent risk measures. 
Further, we provide a simple characterization of robustness for convex risk measures in terms of a weak notion of mixture continuity. 
\end{abstract}

\section{Introduction}
The main goal of this paper is to investigate and characterize risk measures that have ``convex level sets at the level of distributions" (CxLS), in the sense that 
\[
\rho(F)=\rho(G)=\gamma \Rightarrow \rho(\lambda F + (1-\lambda) G) = \gamma, \mbox{ for each } \lambda \in (0,1),
\]
where $F$ and $G$ are probability distributions.
The financial interpretation of this property is that any mixture of two equally risky positions remains with the same risk. In the axiomatic theory of risk measures, it is costumary to impose convexity or quasiconvexity requirements with respect to the pointwise sum of the risks, viewed as random variables on a common state space $\Omega$, in order to model an incentive to diversification. 
On the contrary, the CxLS property arises naturally as a necessary condition for elicitability, that is defined as the property of being the minimizer of a suitable expected loss. More formally, we say that a law determined risk measure $\rho$ is elicitable if there exists a scoring function $S: \R^2 \to \R$ such that
\[
\rho(F)=\arg \min_x \int S(x,y) \diff F (y) .
\]
It has been suggested by several authors that elicitability is a relevant requirement in connection with backtesting and with the comparison of the accuracy of different forecasts of a risk measure. We refer the interested reader to \citet{Gneiting2011}, \citet{Ziegel2013}, \citet{BelliniBignozzi2013}, \citet{EmbrechtsETAL2014}, \citet{Davis2013} and the references therein. \\
In Decision Theory, the CxLS property is usually known as Betweenness; it is one of the possible relaxations of the independence axiom of the Von Neumann-Morgenstern theory \citep[see for example][] {Dekel1986, Chew1989}.
In the seminal paper of \citet{Weber2006}, the author proved that, under additional conditions that we  discuss in detail in Section 3, a monetary risk measure with upper and lower convex level sets at the level of distributions belongs to the class of shortfall risk measures introduced by \citet{FollmerSchied2002} as follows:
\[
\rho(F)= \inf \left\{ m \in \R \mid \int \ell (x-m) \diff F(x) \leq 0\right\},
\]
for a suitable nondecreasing and nonconstant loss function $\ell \colon \R \to \R$. \\
In comparison with Weber's theorem, we  limit ourselves to the more restricted case of convex risk measures. On the contrary, we  relax Weber's additional conditions in order to completely characterize convex law determined risk measures with the CxLS property. We  see in Theorem \ref{characterization} that such risk measures correspond to generalized shortfalls, in which the loss function can also assume the value $+\infty$.
As a consequence, we confirm the result of \citet{Ziegel2013} \citep[see also][]{Weber2006BelliniBignozzi2013}, and show that the only elicitable coherent law determined risk measures are expectiles, that can be defined as the shortfall risk measures associated to the loss function 
\[
\ell_{\alpha}(x)=\alpha x^+ -(1-\alpha) x^-, 
\]
for $\alpha \geq \frac{1}{2}$.
For more information on expectiles we refer to \citet{NeweyPowell1987}, \citet{Delbaen2013}, \citet{BelliniETAL2014}. \\
As a byproduct of our analysis, we provide a simple characterization of robustness for convex risk measures. After the works of \citet{KratschmerETAL2013}, \citet{StahlETAL2012} and others, the notion of robustness for a law determined risk measure is usually identified with its continuity with respect to $\psi$-weak convergence, defined by
\[
F_n \overset{\psi}{\to} F \;\mbox { if } F_n \to F \mbox{ weakly and } \int \psi \diff F_n \to \int \psi \diff F,
\]
where $\psi \colon \R \to [0, +\infty)$ is a continuous gauge function satisfying $\psi \geq 1$ outside some compact set and $\lim_{x\to \infty}\psi(x)=+\infty$. 
We  show in Proposition \ref{WC_prop} that, for a convex law determined risk measure $\rho$, robustness is equivalent to a weak form of mixture continuity:
\[
\lim_{\lambda \to 0^+} \rho (\lambda \delta_x +(1-\lambda) \delta_y) = \rho(\delta_y), \mbox{ for each } x,y \in \R.
\]
It is well known that in the literature on risk measures several systems of notations and sign conventions are coexisting. For ease of reference, in this paper we will follow \citet {Delbaen2012}, so the notion of a convex risk measure will be replaced by that of a concave utility function. \\ 
The paper is structured as follows: In Section 2 we discuss robustness issues, while in Section 3 we provide the characterization of concave, law determined utilities with CxLS. Auxiliary results are moved to the Appendix. 

\section{Continuity properties of law determined concave monetary utility functions}
\subsection{Notations and preliminaries}
In this subsection, we set our notation and review basic properties of concave utility functions. All results that are quoted without reference can be found in \citet{Delbaen2012}. \\
Let $(\Omega,\mathcal{F},\mathbb{P})$ be an atomless probability space. This is not a very big restriction as it simply means that on $\Omega$ we can define a random variable with a continuous distribution function. Several statements are only valid for atomless spaces so we will use this as a standing assumption and will not repeat it.  A utility function is any function $u \colon L^{\infty} \to \R$. 
We say that $u$ is translation invariant if $u(\xi +h)=u(\xi)+h$, for each $h \in \R$; it is monotone if $\xi \leq \eta \mbox{ a.s.} \Rightarrow u(\xi) \leq u(\eta)$; it is monetary, if $u$ is translation invariant, monotone and satisfies $u(0)=0$.  \\
The properties of a monetary utility function can be recovered by means of its acceptance set
\[
\mathcal{A}:= \{ \xi \in L^{\infty} \mid u(\xi) \geq 0 \};
\]
in particular, $u$ is concave if and only if the acceptance set $\mathcal{A}$ is convex. \\
A utility function $u$ is law determined (or law invariant) if 
\[
\operatorname{Law}(\xi)=\operatorname{Law}(\eta) \Rightarrow u(\xi)=u(\eta).
\]
A law determined utility can be seen as a function on $\mathcal{M}_{1,c}(\R)$, the set of probability measures on $\R$ with compact support. The monotonicity property implies that for each $F, G \in \mathcal{M}_{1,c}(\R)$
\[
F \leq_{st} G \Rightarrow u(F) \leq u(G),
\]
where $\leq_{st}$ denotes the usual stochastic order, also known as first order stochastic dominance; if $u$ is concave then also 
\[
F \leq_{cv} G \Rightarrow u(F) \leq u(G),
\]
where $\leq_{cv}$ is the concave order, also known as second order stochastic dominance \citep[see for example] []{BauerleMuller2006,ChernyGrigoriev}. 
The set of distributions of acceptable positions will be denoted by
\[
\mathcal{N}:= \{ \operatorname{Law}(\xi) \mid \xi \in \mathcal{A} \}.
\]
A law determined functional $u \colon \mathcal{M}_{1,c} \to \mathbb{R}$ is weakly continuous if $u(F_n) \to u(F)$ whenever $F_n \to F$ weakly; it is $\psi$-weakly continuous if 
\[
F_n \overset{\psi}{\to} F \,  \Rightarrow u(F_n) \to u(F).
\]
Clearly, since $\psi$-weak convergence implies weak convergence, it follows that 
weak continuity implies $\psi$-weak continuity and that a weakly closed set is $\psi$-weakly closed for each gauge function $\psi$. 
A utility function $u \colon L^{\infty} \to \mathbb{R}$ has the Fatou property if for each $\xi_n \in L^{\infty}$, with
$\sup_n \Vert \xi_n \Vert_{\infty} < +\infty$, it holds that 
\begin{equation}
\label{Fatou}
\xi_n \overset{\mathbb{P}}{\to} \xi \Rightarrow u(\xi) \geq \limsup_{n \to +\infty} u(\xi_n).
\end{equation}
A monetary concave utility function $u \colon L^{\infty }\rightarrow \mathbb{R}$ with the Fatou property has the following dual representation: 
\begin{equation} \label {dual}
u(\xi)=\inf\big\{ \E_{\mathbb{Q}}[\xi] + c(\Q)\; |\;  \Q \in \mathbf{P} \big\},
\end{equation}
where $\mathbf{P}=\{\mathbb{Q} \mid \mathbb{Q} \ll \mathbb{P}\}$ is the set of probability measures that are absolutely continuous with respect to $\mathbb{P}$ and
the penalty function $c\colon \mathbf{P} \to [0, +\infty] $ is convex and lower semicontinuous. We will often identify $\mathbf{P}$ with a subset of $L^1_+$ via the Radon-Nikodym derivative $\diff \mathbb{Q} / \diff \mathbb{P}$. 
If the monetary concave utility function $u$ is law determined, then $u(\xi) \leq \E_{\mathbb{P}} [ \xi ]$ (see Lemma \ref{conditioning} in the Appendix), which
implies that  $c(\mathbb{P})=0$. \\
The following Kusuoka representation holds (\citep{Kusuoka}):
\begin{equation} \label{Kusuoka}
u(\xi ) =\inf\Big\{ \int u_{\alpha }(\xi )\nu (\diff\alpha )+\overline{c}(\nu ) \mid \nu \in \mathcal{M}_{1} [0,1]\Big\},
\end{equation}
where $\mathcal{M}_{1}[0,1]$ is the set of all probability measures on $[0,1]$, $\overline{c}\colon \mathcal{M}_{1}[0,1]\rightarrow [0, +\infty]$ is convex and lower semicontinuous, $u_{\alpha}$ represents the Tail Value-at-Risk  (TVaR) at level $\alpha \in (0,1]$ defined by
\[
u_{\alpha}(\xi)=\frac{1}{\alpha} \int_{0}^{\alpha} q_x(\xi)\diff x,
\] 
where $q_{x}$ denotes a quantile function at level $x$ and $u_{0}(\xi)=\essinf(\xi)$. \\
We say that a monetary concave utility $u$ has the weak compactness (WC) property if the penalty function $c\colon \mathbf{P}\rightarrow [0, +\infty] $ in the dual representation (\ref{dual}) has lower level sets $\mathcal{S}_m :=\{\Q \in \mathbf{P}\mid c(\Q) \leq m \}$ that are compact in the weak topology $\sigma(L^1, L^{\infty})$. 
The WC property is equivalent to the so called Lebesgue property:
\begin{equation} \label{Lebesgue}
\xi _{n}\overset {\mathbb{P}}{\rightarrow} \xi,\; \sup_n\Vert \xi _{n} \Vert_{\infty} < +\infty
\Rightarrow u(\xi _{n})\rightarrow u(\xi ),
\end{equation}
which is a stronger continuity requirement than the Fatou property.
If $u$ is law determined, then the WC property is equivalent to the following property of  
the penalty function $\overline{c}$ in the Kusuoka representation:
\[
\nu (\left\{ 0\right\} )>0 \Rightarrow \overline{c}(\nu )=+\infty.
\]
In the coherent case the dual representation becomes 
\begin{equation*} 
u(\xi)=\inf\big\{ \E_{\mathbb{Q}}[\xi] \mid \Q \in \mathcal{S} \big\},
\end{equation*}
where $\mathcal{S} \subset \mathbf{P}$ is a convex set of probability measures, closed in the $\sigma(L^1, L^\infty)$ topology. The WC property is equivalent to the compactness of $\mathcal{S}$ in the $\sigma(L^1, L^\infty)$ topology, that by the Dunford-Pettis theorem is equivalent to the  uniform integrability of $\mathcal{S}$. Recall that $\mathcal{S} \subset L^1$ is uniformly integrable if 
\[
\forall \varepsilon >0 \,\, \exists \delta >0 \mbox{ such that } \mathbb{P}(A) < \delta \Rightarrow \sup_{\phi \in \mathcal{S}} \int_A \phi \diff\mathbb{P} < \varepsilon.
\]
Uniform integrability is characterized by de la Vall\'ee-Poussin's criterion: $\mathcal{S}$ is uniformly integrable if and only if there exists a function $\Phi\colon \mathbb{R}_+ \to \mathbb{R}$, increasing, convex, with $\Phi(0)=0$ and $\Phi(x)/x \to +\infty$ for $x \to +\infty$, such that 
\[
\sup_{\mathbb{Q} \in \mathcal{S}} \E\Big[\Phi \Big(\frac{\diff\mathbb{Q}}{\diff\mathbb{P}}\Big) \Big] < +\infty.
\]
In the law determined case, the Kusuoka representation becomes
\[
u(\xi)=\inf \left \{  \int \nu( \diff \alpha ) u_{\alpha} (\xi) \mid \nu \in \mathcal{S}   \right \},
\]
and the WC property is equivalent to $\nu(\{0 \}) =0$ for each $\nu \in \mathcal{S}$. \\
A (finite-valued) Young function is a convex function $\Phi:[0,\infty) \to [0,\infty)$ with $\Phi(0) = 0$ and $\lim_{x \to + \infty} \Phi(x) = + \infty$. A Young function is necessarily nondecreasing, continuous and strictly increasing on $\{ \Phi >0 \}$. The Orlicz space $L^{\Phi}$ is defined as
\[
L^{\Phi}:= \{X \;|\; \E[\Phi(c|X|)] < \infty \; \text{for some $c > 0$}\};
\]
the Orlicz heart is 
\[
H^{\Phi}:= \{X \;|\; \E[\Phi(c|X|)] < \infty \; \text{for every $c > 0$}\}.
\]
The Luxemburg norm is defined as
\[
\lVert X \rVert_{\Phi} := \inf \Big\{\lambda > 0 \;|\; \E[\Phi(|X/\lambda|)] \le 1\Big\}
\]
and makes both $L^{\Phi}$ and $H^{\Phi}$ Banach spaces. Finally, we say that $\Phi$ satisfies a $\Delta_2$ condition if $\Phi(2x) \leq k\Phi(x)$, for some $k>0$ and for $x \geq x_0$; in this case $H^{\Phi}$=$L^{\Phi}$. For Orlicz space theory and applications to risk measures we refer to \citet{RaoRen1991}, \citet{EdgarSucheston1992}, \citet{CheriditoLi2008}, \citet{CheriditoLi2009} and the references therein.

\subsection{Mixture continuity properties}

In this subsection we study mixture continuity properties of law determined concave monetary utilities. 
We say that a utility $u$ is mixture continuous if for each $F, G \in \mathcal{M}_{1,c}$, the function
\[
\lambda \mapsto u(\lambda F + (1-\lambda) G)
\]
is continuous on $[0,1]$. The next proposition shows that in the case $F=\delta_x$ and $G=\delta_y$ with $x<y$, continuity for $ \lambda \in (0,1]$ is always satisfied.
\begin{prop} \label{mc1}
Let $u \colon L^{\infty} \to \mathbb{R}$ be a concave law determined monetary utility, and let $x,y \in \R$, with $x <y$. Then the mapping 
\[
\lambda \mapsto u(\lambda \delta_{x}+(1-\lambda) \delta_{y})
\]
is continuous at each $\lambda \in (0,1]$. 
\end{prop}
\begin{proof} 
The $\TVaR$ at level $\alpha$ of the dyadic variables under consideration is given by
\[
u_{\alpha}(\lambda\delta_{x}+(1-\lambda) \delta_{y})=
\begin{cases}
 \frac{\lambda}{\alpha} x +  \frac{\alpha-\lambda}{\alpha} y & \mbox{if } 0 < \lambda < \alpha \\
x & \mbox{if } \alpha \leq \lambda \leq 1
\end{cases},
\]
and is a piecewise linear, non-increasing and convex function of $\lambda$. \\
It follows that for each $\nu \in \mathcal{M}_1[0,1]$, the function 
\[
g_{\nu}(\lambda):= \int u_{\alpha}(\lambda \delta_{x}+(1-\lambda) \delta_{y}) \nu (d\alpha)
\]
is also non-increasing and convex with $g_{\nu}(0)=y$ and $g_{\nu}(1)=x$. \\
By the Kusuoka representation \eqref{Kusuoka}, we have
\[
u(\lambda \delta_{x}+(1-\lambda) \delta_{y})=\inf  \left \{ g_{\nu}(\lambda) + \bar{c}(\nu) \mid \nu \in \mathcal{M}_{1}[0,1] \right \},
\]
and by monotonicity, the function $u(\lambda \delta_{x}+(1-\lambda) \delta_{y})$ is left continuous on $(0,1]$ (at least for $x<y$). Then Lemma \ref{cont} implies that $u(\lambda \delta_{x}+(1-\lambda) \delta_{y})$ is continuous in $\lambda$ for $\lambda \in (0,1)$. Finally, \eqref{Fatou} gives 
\[
\limsup_{\lambda \to 1^-} u(\lambda \delta_x + (1 - \lambda ) \delta_y ) \leq u(\delta_x),  
\] 
so from monotonicity it follows that
\[
\lim_{\lambda \to 1^-} u(\lambda \delta_x + (1 - \lambda ) \delta_y ) = u(\delta_x).
\] 
\end{proof}

The simplest example in which mixture continuity for $\lambda \to 0^+$ fails is the coherent utility $u(\xi)=\essinf(\xi)$. There are also many other examples, such as  $u(\xi)=\gamma \mathbb{E}[\xi] + (1-\gamma) \essinf(\xi)$, for $\gamma \in (0,1)$. In fact, we  prove in Proposition \ref{WC_prop} that for a law determined monetary concave utility the property of mixture continuity for $\lambda \to 0^+$ is equivalent to the WC property. We begin with a characterization of the essential infimum.

\begin{lem} \label{essinf_lemma}
Let $u \colon L^{\infty} \rightarrow \mathbb{R}$ be a concave law determined
monetary utility. If there exists $a > 0$ and $k_n \uparrow 0$ such that $ \forall \alpha \in (0,1)$ we have that 
\[ 
u(\alpha \delta_{k_n} +(1-\alpha) \delta_a)<0,
\]
then $u(\xi)=\essinf(\xi)$.
\end{lem}

\begin{proof}
We  first show that for each $k_n$ and for all $B \in \mathcal{F}$ with $\mathbb{P}(B) > 0$ there exists a $\mathbb{Q} \in \mathbf{P}$ with
\begin{equation}\label{eq:cBc}
c(\mathbb{Q}) \le -k_n \quad \text{and} \quad \mathbb{Q}(B^c) \le \frac{-k_n}{a}.
\end{equation}
Let $B \in \mathcal{F}$ with $\mathbb{P}(B) > 0$. Then from the hypothesis
\[
u(k_n \Ind_{B}+a\Ind_{B^c}) <0
\]
and from the dual representation, there exists $\mathbb{Q} \in \mathbf{P}$ such that 
\[
k_n \mathbb{Q}(B)+a\mathbb{Q}(B^c)+c(\mathbb{Q}) <0,
\]
which yields 
\[
a\mathbb{Q}(B^c)+c(\mathbb{Q})<-k_n \mathbb{Q}(B) \leq -k_n,
\]
hence
\[
c(\mathbb{Q}) \leq -k_n \quad \text{and}\quad \mathbb{Q}(B^c) \leq \frac{-k_n}{a}.
\]
For any $\xi \in L^{\infty}$, $\xi \geq 0$, and $\beta >0$, let 
\[
B:=\{\xi \leq \essinf(\xi) +\beta  \}.
\]
Clearly, $\mathbb{P}(B) > 0$. Taking $\mathbb{Q}$ as in \eqref{eq:cBc}, we have that 
\begin{align*}
u(\xi) \leq E_{\mathbb{Q}}(\xi)+c(\mathbb{Q}) & \leq (\essinf(\xi) + \beta) \mathbb{Q}(B) + \Vert \xi \Vert_{\infty} \mathbb{Q}(B^{c}) + c(\mathbb{Q})\\
& \leq \essinf(\xi) + \beta + \frac{-k_n}{a} \Vert \xi \Vert_{\infty} -k_n.
\end{align*}
Since this inequality holds for each $\beta >0$,  letting $k_n \uparrow 0$ it follows that $u(\xi)=\essinf(\xi)$.
\end{proof}

\begin{lem} \label{WC_lemma}
Let $u \colon L^{\infty} \rightarrow \mathbb{R}$ be a concave law determined
monetary utility. If there exists $a > 0$, $k_n \to -\infty$ and $ \alpha_n \in (0,1)$ such that 
\[ 
u(\alpha_n \delta_{k_n} +(1-\alpha_n) \delta_a) \geq 0,
\]
then $u$ has the WC property.
\end{lem}

\begin{proof}
We show that $\mathcal{S}_m=\left\{ \diff\Q/\diff \p \;|\;c(\Q)\leq m\right\} 
$ is uniformly integrable. Let $\p(A)\leq \alpha_n $. Then $\alpha_n \delta _{k_n}+(1-\alpha_n )\delta_a
\leq _{st} \p(A) \delta_{k_n}+ (1-\p(A)) \delta_{a}$, 
so 
\[
u(k_n\Ind_{A}+a\Ind_{A^{c}})\geq 0, 
\]
that gives for all $\Q$
\[
k_n\Q(A)+a\Q(A^{c})+c(\Q)\geq 0, 
\]
which implies 
\[
\Q(A)\leq \frac{a+c(\Q)}{-k_n} \leq \frac{a+m}{-k_n}.
\]
Letting $k_n \to -\infty$, we find that the set $\mathcal{S}_m$ is uniformly absolutely continuous with respect to $\p$ and hence uniformly integrable. 
\end{proof}

The preceding Lemmas show that there are three possible situations. Let $k<0$ and $a>0$. For the sake of brevity, say that condition C holds for $(k, a)$ if 
there exists $\alpha \in (0,1) \mbox { such that } u(\alpha \delta_k + (1-\alpha) \delta_a ) \geq 0$. 

\begin{prop} 
\label {trichotomy}
Let $u \colon L^{\infty} \rightarrow \mathbb{R}$ be a concave law determined monetary utility. Then there are the following alternatives: 
\begin{itemize}
\item[a)] $u(\xi)=\essinf(\xi)$, in which case condition C does not hold for any $(k,a)$, 
\item[b)] $u(\xi)$ has the WC property, in which case condition C holds for every $(k,a)$, 
\item[c)] $u(\xi) \neq \essinf(\xi)$ and does not have the WC property, in which case condition C holds only for some $(k,a)$.
\end{itemize}
\end{prop}
\begin{proof}
a) If $u(\xi)=\essinf(\xi)$, then clearly condition C never holds. The reverse implication is a straightforward consequence of Lemma \ref{essinf_lemma}. \\
b) If $u(\xi)$ has the WC property, then the function 
\[
\lambda \mapsto u \left (\lambda \delta_{x}+(1-\lambda) \delta_{y} \right )
\]
is continuous also for $\lambda \to 0^+$, as a consequence of the Lebesgue property. Indeed, for any sequence $\lambda_n \to 0^+$, 
let $\xi_n \sim \lambda_n \delta_{x}+(1-\lambda_n) \delta_{y}$, $\xi = y$ $\p$-a.s. and $\xi_n \overset{\mathbb{P}}{\to} \xi$, with $\sup_n\Vert \xi _{n} \Vert_{\infty} \leq \vert y \vert$, so from \eqref{Lebesgue} we get
\[
u(\lambda_n \delta_{x}+(1-\lambda_n) \delta_{y}) \to u(\delta_{y}).
\]
Since $u(k)=k<0$ and $u(a)=a>0$, there exists $\alpha \in (0,1)$ such that $u(\alpha_{\delta_k} + (1-\alpha) \delta_a)=0$; hence condition C holds for every $(k,a)$. The reverse implication is a consequence of Lemma \ref{WC_lemma}. \\
c) follows immediately from a) and b).
\end{proof}
\begin{rem}
It is interesting the comparison with condition (3.1) in \citet{Weber2006}. In our notation, Weber's requirement is that there exists $a_0>0$ such that for each $k<0$, condition C holds for $(k,a_0)$. For concave law determined monetary utilities, Weber's condition (3.1) is satisfied if and only if $u(\xi)$ has the WC property. The if part follows from Proposition \ref{trichotomy} item b), while the only if part follows from Lemma \ref{WC_lemma}. 
\end{rem}

In Proposition \ref{mc1} we showed that any concave law determined monetary utility is mixture continuous on dyadic variables for $\lambda \in (0,1]$. 
From the trichotomy of Proposition \ref{trichotomy} it follows that the additional continuity for $\lambda \to 0^+$ is equivalent to the WC property and to the $\psi$-weak continuity, for some gauge function $\psi$. 

\begin{prop} \label{WC_prop}
Let $u \colon L^{\infty} \rightarrow \mathbb{R}$ be a concave law determined monetary utility. The following are equivalent: 
\begin{itemize}
\item[a)] $u$ has the WC property.
\item[b)] $u$ is $\psi$-weakly continuous for some gauge function $\psi$.
\item[c)] For each $x,y \in \R$ with $x<y$, the function $\lambda \mapsto u(\lambda \delta_{x}+(1-\lambda) \delta_{y})$ is continuous for $\lambda \to 0^+$.
\end{itemize}
\end{prop}

\begin{proof}
First we show that a) $\Rightarrow$ b). Let us first consider the case in which $u$ is coherent. Then, from the WC property it follows that the set of generalized scenarios $\mathcal{S}$ in the dual representation
\[ 
u(\xi)=\inf\big\{ \E_{\mathbb{Q}}[\xi] \mid \Q \in \mathcal{S} \big\}
\]
is compact in the $\sigma(L^1, L^{\infty})$ topology, and hence uniformly integrable. By the de la Vall\'ee-Poussin's criterion, there exists a Young function $\Phi$ with $\frac{\Phi(x)}{x} \to +\infty$ for $x \to +\infty$ such that 
\[
\sup_{\Q \in \mathcal{S}} \E \left [\Phi \left (\frac{\diff \Q}{\diff \p} \right ) \right ]   < +\infty, 
\]
so $\mathcal{S}$ is bounded in the Orlicz space $L^{\Phi}$. Denoting with $\Psi$ the convex conjugate of $\Phi$, from the results of \citet{CheriditoLi2008} and \citet{CheriditoLi2009} it follows that the coherent utility $u$ is finite-valued on the Orlicz heart $H^{\Psi}$. If $\Psi$ satisfies the $\Delta_2$ condition, then from the results of \citet{KratschmerETAL2013} it follows that $u$ is $\psi$-weakly continuous for $\psi(x):=\Psi(\vert x \vert)$, and hence we immediately have b). If instead $\Psi$ does not satisfy the $\Delta_2$ condition, then we consider the gauge function
\[
\psi(x)= \sum_{k=1}^{+\infty} \lambda_k \Psi(k \vert x \vert ), \mbox { where } \lambda_k=\frac{1}{2^k \Psi (k^2)}.
\]
Let $\xi_n \overset{\psi}{\to} \xi$. From Skorohod's representation, it is possible to assume that
\begin{equation}
\label {conv}
\xi_n \to \xi \mbox{ a.s. and }\;  \E[\psi(\vert \xi_n \vert )] \to \E[\psi(\vert \xi \vert )].
\end{equation}
By the continuity of $\psi$, it follows that $\psi(\vert \xi_n \vert) \to \psi(\vert \xi \vert)$ a.s., and since $\psi(\vert \xi_n \vert ) \geq 0$ a.s. and $\E[\psi(\vert \xi_n \vert )] \to \E[\psi(\vert \xi \vert )]$, from Scheff\'e's lemma it follows that $\psi(\vert \xi_n \vert) \overset{L^1} {\to} \psi(\vert \xi \vert)$, so in particular the family $\psi(\vert \xi_n \vert)$ is uniformly integrable. Hence for all $\varepsilon >0$ there exists $C > 0$ such that for all $n$
\[
\varepsilon \geq \E [ \psi( \vert \xi_n \vert ) \Ind_{ \{ \psi( \vert \xi_n \vert ) \geq C \}} ],
\]
and since $\psi(x) \geq \lambda_k \Psi (k \vert x \vert)$ we have that
\[
\varepsilon \geq \lambda_k \E [ \Psi ( k \vert \xi_n \vert ) \Ind_{ \{ \psi ( \vert \xi_n \vert ) \geq C \} } ] \geq \lambda_k \E [ \Psi ( k \vert \xi_n \vert ) \Ind_{ \{ \Psi (k \vert \xi_n \vert ) \geq C / \lambda_k \} } ], 
\]
which yields the uniform integrability of the family
$\Psi (k \vert \xi_n \vert)$. The family $\Psi (k \vert \xi_n - \xi \vert)$ is also uniformly integrable since $\Psi (\vert x-y \vert ) \leq \Psi (2 \vert x \vert) + \Psi (2 \vert y \vert ) $.
Hence, under (\ref{conv}), it holds that 
\[
\E [\Psi (k \vert \xi_n - \xi \vert)] \to 0, \mbox{ for each } k>0,
\]
which in turn implies $\Vert \xi_n -\xi \Vert_{\Psi} \to 0$ (see for example Proposition 2.1.10 in \cite{EdgarSucheston1992}). The thesis follows then from the Young inequality since
\[
\vert u(\xi_n) - u(\xi) \vert \leq \sup_{\Q \in \mathcal{S}} \E_{\Q} [ \vert \xi_n - \xi \vert ] \leq 2 \sup_{\Q \in \mathcal{S}} \left \Vert \frac{\diff \Q}{\diff \p} \right \Vert_{\Phi} \cdot  \Vert \xi_n - \xi \Vert_{\Psi} \to 0.
\] 
Let us now consider the more general case in which $u$ has a dual representation 
\[
u(\xi)=\inf\big\{ \E_{\mathbb{Q}}[\xi] + c(\Q)\; |\;  \Q \in \mathbf{P} \big\}.
\]
From the WC property the sets $\mathcal{S}_k:=\{\Q \in \mathbf{P}\;|\;c(\Q) \leq k\}$ are compact in the $\sigma(L^1,L^{\infty})$ topology, for each $k \geq 0$. 
Let
\[ 
\mathcal{S}_0:=\left \{ \frac{d\Q}{d\p}\frac{1}{1+c(\Q)}~|~ \Q \in \mathbf{P} \right \}.
\]
$\mathcal{S}_0$ is relatively sequentially compact in the $\sigma(L^1, L^{\infty})$ topology, since for any sequence $\Q_n \in \mathcal{S}_0$ there are two alternatives: either $\Q_n$ definitely belongs to some $\mathcal{S}_k$, or for some subsequence $c(\Q_{n_j}) \to +\infty$. In both cases, the sequence $\Q_n$ has a $\sigma(L^1, L^{\infty})$ convergent subsequence. Denoting with $\mathcal{S}$ the solid closed convex hull of $\mathcal{S}_0$, it follows that $\mathcal{S}$ is compact in the $\sigma(L^1, L^{\infty})$ topology and hence uniformly integrable. Let
\[
 E_\mathcal{S}:=\{f\in L^1~|~\exists\,\varepsilon>0, ~\text{s.t. } \varepsilon f\in \sS\}.
\]
$E_\mathcal{S}$ is a Banach space with the Luxemburg norm 
\[
 \|f\|_{\sS}:=\inf\Big\{ \lambda\;\big|\;\frac{|f|}{\lambda}\in \sS\Big\}.
\]
Since $\mathcal{S}$ is uniformly integrable, from the de la Vall\'ee-Poussin's criterion there exists a Young function $\Phi$ with $\frac{\Phi(x)}{x} \to +\infty$ as $x\to +\infty$ such that
\[ 
\sup_{f \in \sS} \E \left [\Phi \left ( |f| \right ) \right ] < +\infty.
\]
Hence
\[
 E_{\sS}\subseteq L^\Phi\subset L^1,
\]
where we can assume $\Phi \in \Delta_2$, so $L^\Phi=H^\Phi$. Passing to the duals we get
\[
 L^\infty \subset L^\Psi \subseteq (E_{\sS})^{\ast},
\]
where $\Psi$ is the convex conjugate of $\Phi$. The dual norm on $(E_{\sS})^{\ast}$ is given by
\[
\Vert \xi \Vert_u =\sup_{\Q \in \p} \E \left [ \frac{\diff \Q}{\diff P} \frac{1}{1+c(\Q)} \vert \xi \vert  \right ],
\]
and since $\Vert \xi \Vert_u \leq 1 \iff u(-\vert \xi \vert) \geq -1$, it follows that $u$ is finite on $(E_{\sS})^{\ast}$ and hence also on the Orlicz heart $H^\Psi \subseteq L^\Psi $. The proof then proceed as in the coherent case. \newline   
To prove that b) $\Rightarrow$ c), let $u$ be $\tilde{\Psi}$-weak continuous and let $x,y \in \R$ with $x<y$. Then
\[
\lambda \delta_x +(1-\lambda) \delta_y \overset{\tilde{\Psi}}{\to} \delta_y \mbox{ for } \lambda \to 0^+,
\]
so from $\tilde{\Psi}$-weak continuity it follows that
\[
u(\lambda \delta_x +(1-\lambda) \delta_y) \to u(\delta_y).
\]
To prove that c) $\Rightarrow$ a), we assume by contradiction that $u$ does not have the WC property. From Proposition \ref{trichotomy} we know that the condition C fails for some $(\bar{k},\bar{a})$, with $\bar{k} <0$ and $\bar{a}>0$. Since from Proposition \ref{mc1} the mapping $\lambda \mapsto u(\lambda \delta_{\bar{k}} + (1-\lambda) \delta_{\bar{a}}))$ is continuous for $\lambda \in (0,1]$, it must hold that 
\[
\lim_{\lambda \to 0^+} u(\lambda \delta_{\bar{k}} + (1-\lambda) \delta_{\bar{a}})) < 0 <\bar{a},
\]
giving a contradiction with c).
\end{proof}

We notice that as a consequence of the Fatou property, it always holds that $F_n \in \mathcal{N}$, $supp (F_n) \subseteq K$ for some compact $K$ and $F_n \to F$ weakly implies that $F \in \mathcal{N}$. From the preceding theorem it follows that if $u$ has the WC property, then the acceptance set $\mathcal{N}$ is $\psi$-weakly closed for some gauge function $\psi$. 
\section{Monetary concave utility functions with CxLS}

From now on we assume that the monetary concave law determined utility $u\colon L^{\infty} \to \mathbb{R}$ has the property of convex level sets at the level of distributions (CxLS), that is
\[
u(F)=u(G)=\gamma \Rightarrow u(\lambda F + (1-\lambda) G) = \gamma, \mbox{ for each } \lambda \in (0,1).
\]
We recall that $\mathcal{N}=\{ \operatorname{Law}(\xi) \mid u(\xi) \geq 0 \}$ is the acceptance set of $u$ at the level of distributions. We have the following:
\begin{lem}\label{intermezzo}
Let $u$ be a monetary concave law determined utility function with CxLS. Then: 
\begin{itemize}
\item[a)] $\mathcal{N}$ and $\mathcal{N}^c$ are convex with respect to mixtures.
\item[b)] Let $u(\xi) \neq \essinf(\xi)$. Then there exists a $k_0 < 0$ such that for each $a>0$ there exists $\alpha \in (0,1)$ such that $u(\alpha \delta_{k_0} +(1-\alpha)  \delta_a) \geq 0$.
\end{itemize}
\end{lem}
\begin{proof}
a) Let $F, G \in \mathcal{N}$ and let $\xi, \eta$ such that $\operatorname{Law}(\xi)=F$ and $\operatorname{Law}(\eta)=G$. Take $A \in \mathcal{F}$ with $\p[A]=\alpha$ and assume the variables $\xi, \eta$ and the set $A$ to be independent. The existence of $\xi$, $\eta$ and $A$ is guaranteed by Lemma \ref{nonatomic}. Without loss of generality, assume that $u(\xi)=u(\eta) +\beta$, with $\beta \geq 0$, and let $\xi ' :=\xi-\beta.$ From translation invariance $u(\xi ')=u(\xi) -\beta = u(\eta)$. Let $F ' =\operatorname{Law}( \xi ')$; then  $\xi \Ind_A + \eta \Ind_{A^c}$ has law $\alpha \F + (1-\alpha) G$ and $\xi ' \Ind_A + \eta \Ind_{A^c}$ has law $\alpha F' + (1-\alpha) G$. Since $\xi \Ind_A + \eta \Ind_{A^c} = \xi ' \Ind_A + \eta \Ind_{A_c} + \beta \Ind_A \geq \xi ' \Ind_A + \eta \Ind_{A_c}$, from monotonicity and CxLS it follows that $u(\xi \Ind_A + \eta \Ind{A^c}) \geq u(\xi ' \Ind_A + \eta \Ind_{A^c})=u(\eta) \geq 0$, that gives the convexity of $\mathcal{N}$ with respect to mixtures. A similar argument applies to $\mathcal{N}
^c$. \\
b) From Proposition \ref{trichotomy}, it follows that there exists $k_0 <0$, $a_0>0$ and $\alpha_0 \in (0,1)$ such that 
$u(\alpha_0 \delta_{k_0} + (1-\alpha_0) \delta_{a_0} ) \geq 0$. We have to prove that for each $a >0$, there exists a suitable $\alpha \in (0,1)$ such that $u(\alpha \delta_{k_0} + (1-\alpha) \delta a) \geq 0$. If $a \geq a_0$, then by monotonicity $\alpha = \alpha_0$ satisfies the thesis. Let then $0 < a < a_0$. Note first that since
$u(\alpha_0 \delta_{k_0} + (1-\alpha_0) \delta_{a_0}) \geq 0$ and $u(0)=0$, from CxLS and a) it follows that for each $\lambda \in (0,1)$ 
\begin{equation*}
u(\lambda \alpha_0 \delta_{k_0} + (1-\lambda) \delta_0 + \lambda (1-\alpha_0) \delta_{a_0} ) \geq 0.
\end{equation*}
By choosing 
\[
\lambda =\frac{a}{(1-\alpha_0)a_0 + \alpha_0 a},
\]
it follows that 
\[
\lambda \alpha_0 \delta_{k_0} + (1-\lambda) \delta_0 + \lambda (1-\alpha_0) \delta_{a_0} \leq_{cv} \lambda \alpha_0 \delta_{k_0} + (1-\lambda \alpha_0) \delta_a,
\]
which implies that $u(\lambda \alpha_0 \delta_{k_0} + (1-\lambda \alpha_0) \delta_a) \geq 0$.
\end{proof}
\begin{rem}
Without the hypothesis of CxLS item b) is false, as can be seen by considering $u(\xi)=\frac{3}{4} \essinf (\xi)  + \frac{1}{4} \E[\xi]$ and $k_0=-a$.
\end{rem}
Lemma \ref{intermezzo} shows that when $u(\xi) \neq \essinf(\xi)$, the quantity
\begin{equation} \label{kappa}
K:=\inf \left\{ k <0  \mid \forall a>0, \, \exists\, \alpha \in (0,1)%
\mbox{ with }u( \alpha \delta _{k}+(1- \alpha )\delta _{a})\geq 0\right\}%
\end{equation}
is well defined, and Proposition \ref{trichotomy} shows that $K=-\infty$ if and only if $u$ has the WC property. 
\begin{lem}
\label{max_lemma}
Let $K$ as in \eqref{kappa}, and let $k \in (K,0)$ and $a>0$. Then 
\[
C(k,a):=\left\{ \alpha \mid u(\alpha \delta _{k}+(1-\alpha )\delta _{a})\geq
0\right\}
\]
is a closed interval. Moreover, letting
\begin{equation}
\label{alpha}
\alpha (k,a):=\max C(k,a),
\end{equation}
it holds that $\alpha(k,a)$ is nondecreasing with respect to $k$ and $a$, and 
\[
u(\alpha (k,a)\delta _{k}+(1-\alpha (k,a))\delta _{a})=0.
\]
\end{lem}
\begin{proof}
The first part of the thesis and the last equality follow from Proposition \ref{mc1}. 
By the assumption $k \in (K,0)$ it follows that $0 < \alpha (k,a) <1$. From the monotonicity of $u$ we have
\[ 
k \leq k', \, a \leq a' \Rightarrow C(k,a) \subseteq C(k ', a'),
\] 
which yields the monotonicity of $\alpha(k,a)$. 
\end{proof}

We now parallel the construction of \cite{Weber2006}, including also the case $K>-\infty$. 
We begin by defining $\varphi \colon (K, +\infty) \to \R$ as in \cite{Weber2006}.
We set $\varphi(0)=0$. 
For $k\in (K,0)$, we define $\varphi (k)$ implicitly by means of 
\[ 
\varphi (k)\alpha (k,1)+(1-\alpha (k,1))=0,
\]
hence
\begin{equation} 
\label {phi1}
\varphi (k)=-\frac{1-\alpha (k,1)}{\alpha (k,1)}=1-\frac{1}{\alpha (k,1)}<0,
\end{equation}
which is nondecreasing in $k$, by Lemma \ref{max_lemma}. For $a>0$, we fix a reference point $k_0 \in (K,0)$ and define $\varphi (a)$ implicitly by means of 
\[
\varphi (k_{0})\alpha (k_{0},a)+\varphi (a)(1-\alpha (k_{0},a))=0,
\]
hence
\begin{equation}
\label {phi2}
\varphi (a)=-\frac{\varphi (k_{0})\alpha (k_{0},a)}{1-\alpha (k_{0},a)}%
=\varphi (k_{0})\left[ 1+\frac{1}{\alpha (k_{0},a)-1}\right] >0,
\end{equation}
which is also nondecreasing in $a$, since $\varphi(k_0) <0$. It can be easily checked that $\varphi (1)=1$, independently on the choice of the reference point $k_0$. 
Thus the affine functional $L_{\varphi} \colon \mathcal{M}_{1,c} \to \R$ given by
\[
L_{\varphi}(\mu)= \int \varphi \diff \mu
\]
with $\varphi$ defined in \eqref{phi1} and \eqref{phi2} vanishes whenever $\mu=\alpha(k_0,a) \delta_{k_0} + (1-\alpha(k_0,a)) \delta_a $ or $\mu=\alpha(k,1) \delta_{k} + (1-\alpha(k,1)) \delta_1$. 
In the following lemma we prove that $L_{\varphi}$ vanishes also on all dyadic variables of the more general form $\mu=\alpha(k,a) \delta_{k} + (1-\alpha(k,a)) \delta_a $, with $\alpha(k,a)$ defined in \eqref{alpha}.
\begin{lem}
\label{Delta_lemma}
Let $K$ be as in \eqref{kappa}, $\alpha(k,a)$ as in \eqref{alpha} and $\varphi$ as in \eqref{phi1} and \eqref{phi2}. Let  $k \in (K,0)$ and $a>0$. Then
\[
\alpha (k,a) \varphi (k) + (1-\alpha (k,a)) \varphi (a)=0. 
\]
\end{lem}

\begin{proof}
If $k=k_0$ or if $a=1$ the thesis is immediate from \eqref{phi1} and \eqref{phi2}, so we assume that $k \neq k_0$ and $a \neq 1$. 
Let 
\[
\Delta =\left\{ ( \lambda _{1},\lambda _{2},\lambda _{3}) \in \R^3 \mid \lambda _{i}\geq 0
\mbox{ and }\sum_{i=1}^3 \lambda _{i}\leq 1\right\},
\]
and let $\Phi \colon\Delta \rightarrow \mathcal{M}_{1,c}(\R)$ be defined by 
\[
\Phi ( \lambda _{1},\lambda _{2},\lambda _{3})=(1-\sum_i \lambda _{i})\delta _{k_{0}}+\lambda
_{1}\delta _{1}+\lambda _{2}\delta _{a}+\lambda _{3}\delta _{k}. 
\]
$\Phi$ is an affine bijective mapping of $\Delta$ onto the finite dimensional face of $\mathcal{M}_{1,c}$ given by the measures with support in $k_0$, $k$, $1$, $a$.
Let 
\begin{align*}
\mathcal{D} &:=\left\{ ( \lambda _{1},\lambda _{2},\lambda _{3}) \in \Delta \mid \Phi ( \lambda _{1},\lambda _{2},\lambda _{3})\in \mathcal{N} \right\} \\
\mathcal{C} &:=\left\{( \lambda _{1},\lambda _{2},\lambda _{3}) \in \Delta \mid \Phi( \lambda _{1},\lambda _{2},\lambda _{3}) \in \mathcal{N}^{c}\right\}.
\end{align*}
From Lemma \ref{intermezzo} it follows that $\mathcal{D}$ and $\mathcal{C}$ are
convex, $\mathcal{D}$ is closed and $\mathcal{C}$ is
relatively open in $\Delta$. We consider the following points in $\Delta$:
\begin{align*}
x^{1}&:=(1-\alpha (k_{0},1),0,0) \\
x^{2}&:=(0,1-\alpha (k_{0},a),0) \\
x^{3}&:=(1-\alpha (k,1),0,\alpha (k,1)) \\
x^{4}&:=(0,1-\alpha (k,a),\alpha (k,a)).
\end{align*}
By \eqref{alpha}, it follows that $x^{i}\in \mathcal{D}$ and each $x^{i}$ is in the relative
closure of $\mathcal{C}$ with respect to $\Delta$. Our aim is to show that $x^{4}$ is an affine combination of $x^1, x^2, x^3$. Since by definition 
$\int \varphi \diff \, \Phi (x^i) =0$ for $i=1, \dots, 3$, this would imply that $\int \varphi \, \diff \Phi (x^4) =0$, which is the thesis. 
In order to show that $x^{4}$ is an affine combination of $x^1, x^2, x^3$, we apply the separation theorem to the convex and disjoint sets $\mathcal{D}$ and $\mathcal{C}$.  
There exists a nontrivial linear $f:\mathbb{R}^{3}\rightarrow \mathbb{R}$ and $s\in \mathbb{R}$ such that $f(x)\leq s$ on $\mathcal{C}$ and $f(x)\geq s$ on $\mathcal{D}$.
Since $(0,0,0) \in \mathcal{C}$ and $f(0,0,0)=0$, it follows that $s\geq 0$. If $s=0$, we would get $f(x^i)=0$ for $i=1, \dots, 3$ since $x^i$ is in the relative closure of $C$, which in turn would imply that $f=0$, which is a contradiction; hence $s>0$. Thus the points $x^1, x^2, x^3, x^4$ lie on the nontrivial hyperplane $f(x)=s$; from the linear independence of $x^1, x^2,x^3$ it follows that $x^4$ is an affine combination of  $x^1, x^2,x^3$.
\end{proof}
\begin{cor}
Let $K$ be as in \eqref{kappa} and $\varphi$ as in \eqref{phi1} and \eqref{phi2}. Let $k, k_0 \in (K,0)$ and $a>0$. If $\xi$ is supported by $k, k_0, 1, a$, then $u(\xi) < 0$ if and only if $\E[\varphi(\xi)] < 0$.
\end{cor}
\begin{lem}
\label{Delta_lemma_2}
Let $K$ be as in \eqref{kappa} and $\varphi$ as in \eqref{phi1} and \eqref{phi2}. Let $\xi$ be supported by finitely many points, all greater than $K$. Then $u(\xi)<0$ if and only if $E[\varphi(\xi)] < 0$. 
\end{lem}
\begin{proof}
Let $k_0, k_1, \dots, k_N < 0 \leq a_1, \dots, a_M$ be distinct points. Similar to the proof of Lemma \ref{Delta_lemma}, we define
\begin{align*}
&\Delta =\{ (\lambda_0, \dots, \lambda_N, \gamma_1, \dots, \gamma_M) \mid \sum \lambda_i + \sum \gamma_j =1, \lambda_i \geq 0, \gamma_j \geq 0 \}, \\
&\Phi(\underline{\lambda}, \underline{\gamma}) =\sum_{i=0}^N \lambda_i \delta_{k_i} + \sum_{j=1}^M \gamma_j \delta_{a_j},
\mathcal{D} = \{ (\underline{\lambda}, \underline{\gamma}) \in \Delta \mid \Phi(\underline{\lambda}, \underline{\gamma}) \in \mathcal{N} \}, \\
& \mathcal{C} = \{ (\underline{\lambda}, \underline{\gamma}) \in \Delta \mid \Phi(\underline{\lambda}, \underline{\gamma}) \in \mathcal{N}^c \}.
\end{align*}
As in the proof of Lemma \ref{Delta_lemma}, $\mathcal{D}$ and $\mathcal{C}$ are convex, $\mathcal{D}$ is closed and $\mathcal{C}$ is relatively open in $\Delta$. By the separation theorem, reasoning as in the proof of Lemma \ref{Delta_lemma}, there is an affine functional $g \colon \Delta \to \R$ such that $g(x) < 0$ for $x \in \mathcal{C}$ and $g(x) \geq 0$ for $x \in \mathcal{D}$.
Let $G:= \{ x \mid g(x)=0  \}$. Then $G \cap \Delta$ is compact and convex and its extremal points lie on the edges of $\Delta$. Let us denote with $\overline{k}_i$ and $\overline{a}_i$ the corners of $\Delta$ corresponding respectively to $\delta_{k_i}$ and $\delta_{a_i}$. There are four type of edges: $[\overline{k}_i, \overline{k}_{i'} ]$, $[\overline{a}_j, \overline{a}_{j'} ]$, $[\overline{k}_i, 0 ]$, $[\overline{k}_i, \overline{a}_{j} ]$, with ${a}_{j} >0$. Analyzing the different possibilities, it follows that the extremal points of $G \cap \Delta$ corresponds to $\delta_0$ or to dyadic variables of the form $\alpha(k,a) \delta_k +(1-\alpha) (k,a) \delta _a$, with $\alpha (k,a)$ given by \eqref{alpha}. It follows that $\int \varphi \diff \Phi (x)=0$ for all extremal points of $G \cap \Delta$, and hence $\int \varphi \diff \Phi (x)=0$ for each $x \in G \cap \Delta$. Let $F= \{ x \in \R^{N+M+1} \mid \int \varphi \diff \Phi (x)=0, \, \sum_{i=0}^{N+M}x_i= 1 \}$. $F$ and $G$ are affine spaces of dimension $N+M-1$ and $G \subset F$, so that $G=F$. Hence $x \in C \iff g(x) < 0 \iff \int \varphi \diff \Phi (x) < 0$, which gives the thesis.

\end{proof}

\begin{lem}
\label{concavity_lemma}
The function $\varphi \colon (K, +\infty) \to \mathbb{R}$ is concave on $(K,0)$ and on $[0, +\infty)$ and right continuous.
\end{lem}
\begin{proof}
Take $k_1, k_2 <0$ and let $\alpha_1:=\alpha(k_1,1)$ and $\alpha_2:= \alpha(k_2,1)$, so that 
\[
u(\alpha_1 \delta_{k_1} + (1-\alpha_1) \delta_1)= u(\alpha_2 \delta_{k_2} +(1-\alpha_2) \delta_1) =0.
\]
From CxLS, for each $\gamma \in (0,1)$, it holds that 
\[
u( \gamma (\alpha_1 \delta_{k_1} + (1-\alpha_1) \delta_1 ) + (1-\gamma) (\alpha_2 \delta_{k_2} +(1-\alpha_2) \delta_1 ) )=0,
\]
or equivalently
\[
u \left ( \gamma  \alpha_1 \delta_{k_1} + (1-\gamma) \alpha_2 \delta_{k_2} + \left [ \gamma (1-\alpha_1)  + (1-\gamma) (1-\alpha_2) \right ] \delta_1 \right )=0.
\]
Let 
\[ 
\lambda = \frac{\gamma \alpha_1}{\gamma \alpha_1 + (1-\gamma) \alpha_2}
\]
and $k=\lambda k_1 + (1-\lambda) k_2$.
Then
\[
\gamma  \alpha_1 \delta_{k_1} + (1-\gamma) \alpha_2 \delta_{k_2} \leq_{cv} (\gamma \alpha_1 + (1-\gamma) \alpha_2) \delta_k,
\]
hence the isotonicity of $u$ with respect to the concave order implies
\[
u \left ( (\gamma \alpha_1 + (1-\gamma) \alpha_2) \delta_k  + \left [ \gamma (1-\alpha_1)  + (1-\gamma) (1-\alpha_2) \right ] \delta_1 \right ) \geq 0,
\]
and from Lemma \ref{Delta_lemma_2} we get
\[
(\gamma \alpha_1 + (1-\gamma) \alpha_2) \varphi(k) + \left [ \gamma (1-\alpha_1)  + (1-\gamma) (1-\alpha_2) \right ] \geq 0,
\]
or
\begin{align*}
\varphi (k) &\geq - \frac{(1-\gamma)(1-\alpha_2) + \gamma (1- \alpha_1) }{\gamma \alpha_1 +(1-\gamma) \alpha_2}  \\
&=\frac{\gamma \alpha_1}{\gamma \alpha_1 + (1-\gamma) \alpha_2 } ( - \frac{1-\alpha_1}{\alpha_1}) + 
\frac{(1-\gamma) \alpha_2}{\gamma \alpha_1 + (1-\gamma) \alpha_2 } ( - \frac{1-\alpha_2}{\alpha_2})  \\
&= \lambda \varphi(k_1) + (1-\lambda) \varphi (k_2).
\end{align*}
A similar argument can be used to establish concavity on $[0, +\infty)$. \\
Let $a_n \downarrow 0$ and let $\alpha_n:= \alpha(k_0, a_n)$, with $\alpha$ given by \eqref{alpha}. From the monotonicity of $\alpha$, the sequence $\alpha_n$ is nonincreasing; we denote with $\alpha_0$ its limit. Let $B_{\alpha_n} \in \F $ with $ \p(B_{\alpha_n})=\alpha_n$, and let  $\xi_n=k_0 \Ind_{B_{\alpha_n}} + a_n \Ind_{B_{\alpha_n}^c}$. Then $\xi_n \overset{\p}{\to} \xi$, with $\xi=k_0 \Ind_{B_{\alpha_0}} $, and from the Fatou property $u(\xi) \geq \limsup u(\xi_n) \geq 0$, since $\xi_n \in \mathcal{A}$. Hence $\E_{\p}[\xi] = k_0 \p [B_{\alpha_0}] \geq u(\xi) \geq 0$, which gives $\alpha_0=0$. 
Since $\varphi(a_n)= \frac{-\alpha_n \varphi(k_0)}{1-\alpha_n}$, it follows that $\varphi(a_n) \to 0$ when $a_n \to 0$. 
\end{proof}
From now on we define
\[
\varphi(K)=\lim_{x \downarrow K} \varphi(x).
\]
\begin{lem}
\label{xi_lemma}
If $\xi \geq K$ a.s., then $u(\xi) < 0$ if and only if $E[\varphi(\xi)] <0$. 
\end{lem}
\begin{proof}
Let $\xi_n > K$, $\xi_n$ finitely supported, with $\xi_n \downarrow \xi$ and $\Vert \xi_n - \xi \Vert_{\infty} \to 0$. Then
\[
u(\xi) < 0 \iff \exists n \mbox{ s.t. }u(\xi_n) <0 \iff \exists n \mbox{ s.t. } \E[\varphi(\xi_n)] <0 \iff \E[\varphi(\xi)] <0,
\]
where the second equivalence follows from Lemma \ref{Delta_lemma_2} and the last equivalence from the right continuity of $\varphi$ in $0$.
\end{proof}
\begin{lem}
$\varphi$ is concave on $(K, +\infty)$ and hence continuous on $(K, +\infty)$. 
\end{lem}
\begin{proof}
The proof proceeds along the lines of Lemma \ref{concavity_lemma}, using Lemma \ref{xi_lemma} instead of Lemma \ref{Delta_lemma_2} to have the stronger thesis. 
\end{proof}
We can finally prove the announced characterization:
\begin{thm} \label{characterization}
Let $u\colon L^{\infty} \to \R$ be a monetary, concave law determined utility with CxLS. Then there exists a concave $\overline {\varphi} \colon \mathbb{R} \to \mathbb{R} \cup \{ -\infty \}$ such that $u(\xi) \geq 0$ if and only if $E[\overline {\varphi} (\xi)] \geq 0$.
\end{thm}
\begin{proof}
If $u(\xi)=\essinf(\xi)$, then 
\[
\overline{\varphi} (x)=
\begin{cases}
  - \infty & \mbox{ if } x<0 \\
 \mbox{   } 0 & \mbox { if } x \geq 0
\end{cases}
\]
satisfies the thesis. 
If $u(\xi) \neq \essinf(\xi)$, define
\[
\overline{\varphi}=
\begin{cases}
  - \infty & \mbox{ if } x<K \\
\varphi(x)  & \mbox { if } x \geq K
\end{cases}
\]
with $\varphi$ and $K$ defined as before. If $\xi \geq K$ a.s., then Lemma \ref{xi_lemma} gives that $u(\xi) \geq 0$ if and only if $E[\varphi(\xi)] \geq 0$, which is the thesis. If $\p(\xi < K) >0$, then $\E[ \overline{\varphi}(\xi)] =-\infty$. In order to show that $u(\xi) < 0$, let $\mathcal{B}$ be the algebra generated by the events $\{ \xi < K \}$ and $\{ \xi \geq K \}$, and let $\eta = \E [\xi \Ind_{\xi < K} + \xi \Ind_{\xi \geq 0} | \mathcal{B}]$. Then 
\[
u(\eta) \geq u(\xi \Ind_{\xi < K} + \xi \Ind_{\xi \geq 0}) \geq u(\xi),
\]
and $u(\eta) <0$ by definition of $K$, so that $u(\xi) <0$, which completes the proof. 
\end{proof}

\subsection{Examples}
\begin{exmp}[Essential infimum]
Let $\varphi \colon \R \to \R \cup \{ -\infty \}$, with
\[
\varphi(x)=
\begin{cases}
  - \infty & \mbox{ if } x<0 \\
 \mbox{   } 0 & \mbox { if } x \geq 0
\end{cases}
\]
Then $\A=\{ \xi | \xi \geq 0 \}$ and $u(\xi)=\essinf(\xi)$.
\end{exmp}
\begin{exmp}[Finite Shortfall]
Let $\varphi \colon \R \to \R$ concave and increasing with $\varphi(0)=0$. Then 
\[
\mathcal{A}= \{ \xi \in L^{\infty} | \E [ \varphi(\xi)] \geq 0 \}
\]
is the acceptance set of a concave law determined utility $u$ with CxLS. In the particular case $\varphi (x)=x$ we have $u(\xi)= \E[\xi]$. Let us remark that the function  $\varphi(x)=x$ for $x\leq 0$ and  
$\varphi(x)=0$ for $x\geq 0$ also defines the essential infimum.
\end{exmp}
\begin{exmp}[Truncated shortfall]
\label{truncated}
Let $\varphi \colon \R \to \R$ concave and increasing with $\varphi(0)=0$ and $K<0$. Set $\overline{\varphi} \colon \R \to \R \cup \{ -\infty \}$
\[
\overline{\varphi}(x):=
\begin{cases}
 - \infty & \mbox{ if } x<K \\
\varphi(x)  & \mbox { if } x \geq K
\end{cases}
\]
Then $\A =\{  \xi | \E[\varphi(\xi) \geq 0], \xi \geq K \}$.
\end{exmp}
\begin{exmp}[Truncated mean]
Let 
\[
\overline{\varphi}(x):=
\begin{cases}
x  & \mbox { if } x \geq -1 \\
 - \infty & \mbox{ if } x<-1 
\end{cases}
\]
Then 
\[
\mathcal{A}=\left\{ \xi \in L^{\infty }\text{ s.t. }\xi \geq -1\text{ and }%
\E[\xi ]\geq 0\right\},
\]
and  
\[
u(\xi )=\min (\E[\xi ],1+\essinf (\xi )).
\]
It is easy to see that mixture continuity for $\lambda \to 0^+$ fails, so from Proposition \ref{WC_prop} the WC property does not hold. Indeed the penalty function $c$ is given by 
\[
c(\Q)=1-\essinf \frac{\diff \Q}{\diff \p}
\]
and does not have $\sigma(L^1, L^{\infty})$ compact lower level sets, since $c(\Q) \leq 1$.
\end{exmp}
\begin{exmp}
Let $K<0$ and 
\[
\varphi(x)=
\begin{cases}
\log (x-K)-\log(-K) &   \mbox{ if } x>K \\
 - \infty  & \mbox { if } x \leq K
\end{cases}
\]
\end{exmp}
\begin{exmp}
\[
\varphi(x):=
\begin{cases}
\sqrt{x+1}& \mbox{ for } x \ge -1 \\
- \infty& \mbox{ elsewhere}.
\end{cases}
\]
\end{exmp}
Note that in the last two examples $\varphi^{\prime}_+(K) = +\infty$, so they do not belong to the family of ``truncated shortfalls'' considered in Example \ref{truncated}.

\section{Appendix}
\begin{lem} 
\label{nonatomic}
Let $(\Omega, \mathcal{F}, \mathbb{P})$ be an atomless space. Then it is possible to construct two increasing families of measurable subsets $(A_t)_{0 \leq t \leq 1}$ and $(B_t)_{0 \leq t \leq 1}$, such that for each $t \in [0,1], \, P(A_t)=P(B_t)=t$, and the sigma algebras $\mathcal{A}:=\sigma(A_t; 0 \leq t \leq 1)$ and $\mathcal{B}:=\sigma(B_t; 0 \leq t \leq 1)$ are independent.
\end{lem}
\begin{proof}
See \cite {Delbaen2012}.
\end{proof}
\begin{lem}
\label{conditioning}
If $u \colon L^{\infty} (\Omega, \mathcal{F}, \p) \to \mathbb{R}$ is a concave and law determined utility, then for each $\mathcal{G} \subset \mathcal{F}$ and $\xi \in L^{\infty}$ it holds that 
$u \left ( \E \left [ \xi | \mathcal{G} \right ] \right ) \geq u(\xi)$.
\end{lem}
\begin{proof}
Let start with the case $\mathcal{G}=\{ \emptyset, \Omega\}$. Take $\xi_n$ i.i.d. with $\operatorname{Law}(\xi_n)=\operatorname{Law}(\xi)$, so that $u(\xi_n)=u(\xi)$. From the strong law of large numbers, 
\[
\frac{\xi_1+ \dots + \xi_n}{n} \overset{a.s.}{\to} \E[\xi],
\]
and from the Fatou property and the concavity of $u$
\[
u(\E[\xi]) \geq \limsup_{n \to +\infty} u \left( \frac{\xi_1+ \dots + \xi_n}{n} \right ) 
\geq \limsup_{n \to +\infty} \frac{1}{n} \sum_{k=1}^n u( \xi)=u(\xi).
\]
Let now $\mathcal{G}$ be generated by a finite partition $G_1, \dots, G_m$. Since $\Omega$ is atomless, it is possible to construct a sequence $\xi_n$ that is i.i.d. for each conditional probability $\p[ \cdot |G_k]$, and such that $\operatorname{Law}[ \xi_n | G_k] =  \operatorname{Law}[ \xi | G_k]$, for each $k=1, \dots, m$. The previous reasoning can be applied to each $G_k$, and since 
\[
\frac{\xi_1 + \dots + \xi_n }{n} \overset{a.s.}{\to} \E[\xi | \mathcal{G}],
\]
it follows that $u(\E[\xi|\mathcal{G}]) \geq u(\xi)$.
Finally, for a general $\mathcal{G}$ and for a fixed $\xi \in L^{\infty}$, it is always possible to construct a sequence of finitely generated $\mathcal{G}_n$ such that 
\[
\E[\xi | \mathcal{G}_n] \overset{L^{\infty}}{\to} \E [ \xi | \mathcal{G}],
\]
so we have 
\[
u(\E[\xi | \mathcal{G}]) = \lim_{n \to +\infty} u( \E[ \xi | \mathcal{G}_n] ) \geq u(\xi).
\]
\end{proof}

\begin{lem} \label{cont}
Let $\phi_i \colon [0,1] \to \mathbb{R}$ be nondecreasing and convex with $D_0 < \phi_i < D_1$, and let $c_i$ be bounded from below. Then $\phi(x):=\inf_{i}\big(\phi_i(x) + c_i\big)$ is Lipschitz on compact subsets of $(0,1)$. 
\end{lem}
\begin{proof}
Let $0 <\varepsilon <1$, $\varepsilon \leq x < y \leq 1-\varepsilon$, and $z=y+\varepsilon$. Then $z \leq 1$ and
\[
y=(1-\lambda) z + \lambda x, \mbox { with } \lambda=\frac{\varepsilon}{\varepsilon+ y - x}.
\]
Hence $\phi_i(y) \leq (1-\lambda) \phi_i(z) + \lambda \phi_i(x)$, that gives
\[
\phi_i(y)-\phi_i(x) \leq (1-\lambda) [\phi_i(z) - \phi_i(x)] \leq \frac{y-x}{\varepsilon} (D_1-D_0).
\]
It follows that the family $\{ \phi_i(x) + c_ i \}_i $ is equi-Lipschitz on $[\varepsilon, 1- \varepsilon]$, which implies that also $\phi(x)=\sup_{i}\big(\phi_i(x) + c_i\big)$ is Lipschitz on $[\varepsilon, 1- \varepsilon]$; letting $\varepsilon \to 0$ gives the thesis.
\end{proof}

\newpage

\bibliographystyle{plainnat}
\bibliography{biblio}

\end{document}